\documentclass[12pt,a4paper]{article}
\usepackage[utf8]{inputenc}
\usepackage{eqnarray,amsmath}
\usepackage{amsmath}
\usepackage{amsthm}
\usepackage{amsfonts}
\usepackage{amsmath,bm,bbm}
\usepackage{multirow,array}
\usepackage{amssymb}
\usepackage{graphicx}
\usepackage{quoting}
\quotingsetup{font=small}
\usepackage[left=2cm,right=2cm,top=3cm,bottom=3cm]{geometry}
\usepackage[italian,english]{babel}
\usepackage{setspace}
\usepackage[authoryear,round]{natbib}
\usepackage{lipsum}
\usepackage[nodvipsnames]{color}
\usepackage{verbatim}
\usepackage{natbib}
\usepackage{tikz}
\usetikzlibrary{arrows.meta}
\usetikzlibrary{matrix,arrows}
\usepackage{subfig}
\usepackage{xcolor}
\usepackage{natbib} 
\usepackage{booktabs}  

\usepackage{enumerate}

\usepackage{url}					
\usepackage[pagebackref=false, linktoc=all]{hyperref}
\hypersetup{colorlinks = true, unicode=true,  urlcolor=blue, linkcolor=blue, filecolor=blue, citecolor=blue}    

\usepackage{verbatim}

\usetikzlibrary{calc}

\newtheorem{theorem}{Theorem}

\newtheorem{lemma}[theorem]{Lemma}
\newtheorem{proposition}[theorem]{Proposition}

\newtheorem{definition}{Definition}

\begin{document}

    \title{Truthful Communication and Exclusive Information Clubs}

\author{Paolo Pin\thanks{Department of Economics, Universit\`a di Siena \& BIDSA, Universit\`a Bocconi. Email: paolo.pin@unisi.it}}
	
\maketitle
	
\vspace*{-\baselineskip}

\begin{abstract}
This paper studies how the possibility of strategic misreporting shapes endogenous
communication networks. Agents observe noisy private signals about a common state,
form costly communication links, exchange private messages with their neighbors, and
then choose actions. Payoffs reward both accuracy and coordination with linked agents.
A link is valuable because it gives access to information, but it is useful only if the
induced local information structure makes truthful transmission incentive compatible.
We show that clique components support truthful communication: within a clique, all
members observe the same profile of local signals, choose the same posterior action, and
therefore have no incentive to distort reports. With heterogeneous signal precisions and
convex linking costs, the core selects assortative information clubs ordered by signal
precision. These stable truthful networks need not be socially efficient. Because the
informational value of precision is decreasing, concentrating high-precision agents in the
same club may be privately stable but socially dominated by more mixed partitions.
\end{abstract}


\bigskip


\section{Introduction}
\label{sec:intro}

Why do informed agents sometimes communicate in closed groups? In many economic
environments, information is valuable only if it can be transmitted credibly. Researchers
exchange preliminary ideas with selected colleagues; firms share information within
restricted alliances; traders and professionals rely on trusted circles rather than open
communication. These arrangements look like exclusive information clubs: links are
formed not only to access signals, but also to preserve the incentives that make signals
truthfully transmissible. This paper studies this mechanism in a networked communication
game with endogenous links.

Agents form a communication
network, observe noisy private signals about a common state,  exchange private messages with their neighbors, and then choose actions.
Payoffs reward both accuracy with respect to the state and coordination with linked
agents. A link is therefore valuable for two reasons. First, it gives access to another
agent's signal, improving inference about the state. Second, it can improve coordination
by aligning actions across linked agents. But links also create a strategic problem. Since
messages affect the receiver's beliefs and actions, a sender may prefer to distort her
report in order to move the receiver's action closer to her own. Thus, communication
links are not simply information pipes: their value depends on whether the local
information structure induced by the network makes truthful reporting incentive
compatible.

The first contribution of the paper is to identify a simple architecture that solves this
problem. Clique components support truthful communication. If a connected component
is a clique, then, under truthful reporting, every member observes the same profile of
local signals. Hence all members have the same posterior information and choose the
same action. In such an environment, a unilateral misreport does not give the sender a
useful informational advantage over the receiver; it only moves the receiver's action away
from the action induced by the common information set. Truthful reporting is therefore
sequentially optimal on every link inside the clique. In this sense, a clique is not only a
dense communication group. It is an institutional arrangement that makes truthful
communication self-enforcing.

The second contribution is to study how such truthful communication structures can
arise endogenously. Agents differ in the precision of their private signals. High-precision
agents are more valuable communication partners, because their signals generate larger
reductions in posterior uncertainty. At the same time, the marginal value of adding a new
signal to a clique decreases with the precision already accumulated by that clique. With
convex linking costs, these two forces generate exclusive information clubs. Starting from
the most precise agents, cliques expand as long as the marginal informational value of an
additional member exceeds the marginal linking cost. Once the next agent is not worth
adding, the clique closes, and the next club is formed by the most precise agents who
remain outside. The resulting stable truthful architecture is an assortative partition into
cliques ordered by signal precision.

This perspective separates three notions that are often conflated: truthful communication,
strategic stability, and social efficiency. Clique architectures are attractive because they
support truthful communication and can be stable under decentralized link formation.
They need not be socially efficient. When an agent joins a clique, her signal benefits all
members of the clique, but individual linking decisions do not fully internalize this
informational externality. Stable clubs may therefore be too small. Moreover, because
information has diminishing marginal value, concentrating high-precision agents in the
same club may be privately stable but socially dominated by more mixed partitions that
spread precise signals across groups. Thus, the same force that sustains truthful
communication can generate inefficient informational segregation.

The paper relates first to the literature on Gaussian information aggregation and
coordination. In models such as \citet{MyattWallace2011}, agents use noisy signals and
equilibrium actions are linear functions of available information. Our continuation game
has the same Gaussian-quadratic structure once communication is truthful. The key
difference is that truthful information sharing is not imposed. It is an equilibrium
property, and its feasibility depends on the network. This distinction matters because the
relevant information structure is local: agents do not observe all signals, and incentives
to reveal a signal depend on how the sender's and receiver's information sets compare.

Second, the paper relates to strategic communication in networks. In
\citet{GaleottiGhiglinoSquintani2013}, communication affects beliefs and actions through
the network, and agents may have incentives to distort information. We share the view
that network position matters for communication incentives, but ask a different question:
which endogenous network architectures make truthful local communication sustainable?
The answer developed here is that clique components align local information sets and
therefore remove the informational asymmetry that creates distortion incentives.

Third, the paper contributes to the literature on endogenous network formation and
information exchange. \citet{BillandBravardKamphorstSarangi2017} study networks
formed to obtain and confirm information, while \citet{KushnirNichifor2014} analyze how
agents form connections and then choose how to share private information. These papers
show that informational benefits shape stable network architectures. The present paper
adds a strategic-transmission constraint: a link is valuable only if the information it
carries can be transmitted truthfully in the continuation game.

The paper is also related to reduced-form models of network formation with
component-level externalities. \citet{erol_vohra_2022} study a setting in which
agents' payoffs depend on their own degree and on the number of links in their
connected component. They show that negative link externalities lead to disjoint
cliques, while positive link externalities can support complete or star networks.
Our mechanism is different. In the present model, clique formation is not driven
by reduced-form externalities from links in the component, but by the incentive
compatibility of truthful information transmission: cliques equalize local
information sets and thereby remove incentives to distort reports.

The paper is also connected to models of communication on fixed information networks.
\citet{deMartiMilan2019} study a regime-change game in which agents truthfully share
private signals with their neighbors, so that the network determines the precision and
correlation of posterior beliefs. Here, by contrast, truthful communication is endogenous:
the network must not only transmit information, but also support incentives for truthful
reporting. This shifts attention from the consequences of a given information network to
the compatibility between network formation and strategic communication.

Finally, the network-formation mechanism is related to work in which quadratic
coordination motives generate endogenous social structure. In
\citet{bolletta_pin_2025}, agents choose links and update opinions dynamically; complete
components emerge because agents value coordination with similar peers. The present
paper studies a static information-transmission counterpart. Clique components arise not
only because they reduce disagreement, but because they equalize local information sets
and thereby sustain truthful communication. The mechanism is also close in spirit to the
assortative-clique logic in \citet{allmis_pin_vegaredondo_2026}. There, agents sort into
ordered cliques according to their ex--ante probability of becoming informed. Here, the
ranking variable is signal precision, and clique formation is driven by the need to make
truthful information transmission incentive compatible.

The rest of the paper is organized as follows. Section \ref{sec:model} introduces the
Gaussian communication game, the timing, and the network-formation environment with
convex linking costs. Section \ref{sec:truthful} studies the continuation game and proves
that clique partitions support truthful communication. Section \ref{sec:formation}
derives the assortative clique-formation result and discusses efficiency. Section
\ref{sec:conclusion} concludes.
Appendix \ref{app:proofs} contains the formal proofs of all the results.

\section{Model}
\label{sec:model}

Let \(N=\{1,\ldots,n\}\), with \(n\geq 2\), be a finite set of agents. Agents face
an unknown state
\[
\theta\sim \mathcal N(0,\sigma_\theta^{-1}),
\qquad \sigma_\theta>0.
\]

A network is an undirected graph
\[
g\subseteq \{\{i,j\}:i,j\in N,\ i\neq j\}.
\]
We write \(ij\in g\) when agents \(i\) and \(j\) are linked. Agent \(i\)'s set of
neighbors is
\[
N_i(g)=\{j\in N:ij\in g\},
\]
and her closed neighborhood is
\[
\bar N_i(g)=N_i(g)\cup\{i\}.
\]
The degree of agent \(i\) is denoted by
\[
d_i(g)=|N_i(g)|.
\]

Links are costly. If agent \(i\) maintains \(d_i(g)\) links, she pays
\(c(d_i(g))\), where
\[
c:\{0,\ldots,n-1\}\to \mathbb R_+.
\]
We normalize \(c(0)=0\). The cost function is increasing and weakly convex in the
discrete sense. Let
\[
\Delta c(d)=c(d+1)-c(d).
\]
We assume
\[
\Delta c(d)\geq 0
\quad\text{and}\quad
\Delta c(d+1)\geq \Delta c(d)
\]
whenever these expressions are well defined. Thus the marginal cost of links is
weakly increasing in degree. Strict convexity is not imposed.

The timing is as follows.

\begin{enumerate}
    \item \textbf{Network formation.}
    Agents form an undirected communication network \(g\). Link formation requires
    mutual consent and is governed by the stability notion defined below.

    \item \textbf{Information.}
    Each agent \(i\) observes a private signal
    \[
    x_i=\theta+\eta_i,
    \qquad
    \eta_i\sim \mathcal N(0,\tau_i^{-1}),
    \]
    where \(\tau_i>0\). Signal errors are independent across agents and independent
    of \(\theta\). The precision profile
    \[
    \tau=(\tau_1,\ldots,\tau_n)
    \]
    is exogenous and common knowledge.

    \item \textbf{Communication.}
    Agents communicate locally. Each agent \(i\) sends a private message to each
    neighbor \(j\in N_i(g)\). A message strategy for agent \(i\) is a measurable
    function
    \[
    m_i:\mathbb R\to \mathbb R^{|N_i(g)|},
    \]
    where \(m_{ij}(x_i)\) is the message sent from \(i\) to \(j\). Messages are
    private: agent \(j\) observes \(m_{ij}(x_i)\), while agents other than \(j\) do
    not.

\item \textbf{Actions.}
After observing her own signal, the messages she sends, and the messages received
from her neighbors, each agent chooses an action \(a_i\in\mathbb R\). An action
strategy is a measurable function
\[
\alpha_i:\mathbb R \times \mathbb R^{|N_i(g)|}\times \mathbb R^{|N_i(g)|}\to \mathbb R,
\]
where \(\alpha_i(x_i,m_i^{out},m_i^{in})\) maps the agent's own signal,
the vector of messages sent to her neighbors, and the vector of messages
received from her neighbors into an action.

    \item \textbf{Payoffs.}
    Payoffs are realized after actions are chosen. If \(d_i(g)>0\), agent \(i\)'s
    payoff is
    \[
    U_i(a,\theta,g)
    =
    -(a_i-\theta)^2
    -
    \frac{\gamma}{d_i(g)}
    \sum_{j\in N_i(g)}(a_i-a_j)^2
    -
    c(d_i(g)).
    \]
    If \(d_i(g)=0\), the coordination term is defined to be zero, so that
    \[
    U_i(a,\theta,g)=-(a_i-\theta)^2.
    \]
    The parameter \(\gamma>0\) measures the importance of coordination with linked
    agents. The first term rewards accuracy with respect to the state, the second
    term rewards coordination with neighbors, and the last term is the cost of
    maintaining links.
\end{enumerate}

We use the following equilibrium notions.

\begin{definition}[Continuation game]
Fix a network \(g\) and a precision profile \(\tau\). The continuation game induced by
\((g,\tau)\) consists of the information, communication, and action stages that follow
network formation.
\end{definition}

Throughout, continuation equilibria are Perfect Bayesian equilibria.

\begin{definition}[Truthful continuation equilibrium]
A continuation equilibrium of the game induced by \((g,\tau)\) is truthful on \(g\) if,
for every link \(ij\in g\),
\[
m_{ij}(x_i)=x_i
\]
for all signal realizations \(x_i\).
\end{definition}

\begin{definition}[Core]
An outcome, consisting of a network and continuation strategies, is in the core if there
is no coalition \(B\subseteq N\) and no feasible deviation by the members of \(B\) to an
alternative network and continuation strategies such that all agents in \(B\) are weakly
better off and at least one agent in \(B\) is strictly better off.
\end{definition}

All comparisons in the core are in terms of ex--ante expected payoffs, before signals are
realized. A deviation by a coalition may involve both changing the links among its members
and changing the continuation strategies used within the deviating coalition.

We solve the model by backward induction. First, for a fixed network \(g\), we study the
continuation game and identify network structures that sustain truthful communication.
Second, we use these truthful continuation outcomes to analyze which communication
structures survive coalitional deviations.

\section{Continuation Equilibria and Truthful Communication}
\label{sec:truthful}

We study the continuation game for a fixed network \(g\) and precision profile
\(\tau\). The network has already been formed. The question is which communication
equilibria arise and, in particular, when truthful communication can be sustained.

\subsection{Babbling}

\begin{proposition}[Babbling equilibrium]
\label{prop:babbling}
For every network \(g\) and precision profile \(\tau\), the continuation game
admits an uninformative (babbling) equilibrium.
\end{proposition}

The proposition implies that the continuation game always admits an uninformative
equilibrium. In particular, whenever another equilibrium is feasible, the continuation
game exhibits multiplicity.

\subsection{Linear equilibrium under truthful communication}

We next fix a network \(g\) and ask what happens in the action stage if all
messages are truthful. In that case, agent \(i\)'s information consists exactly of
her own signal and the signals reported by her neighbors, that is, the signals in
her closed neighborhood \(\bar N_i(g)\).

Because payoffs are quadratic and signals are Gaussian, the induced action game has
a linear structure. The following lemma records the resulting equilibrium action
profile, which will be used to study incentives for truthful reporting.

\begin{lemma}[Linear equilibrium under truthful communication]
\label{lem:linear_truthful}
Fix a network \(g\) and a precision profile \(\tau\). Under truthful communication,
the action stage admits a unique Bayesian equilibrium. This
equilibrium is linear: each agent \(i\)'s action is of the form
\[
a_i=\sum_{\ell\in \bar N_i(g)} b_{i\ell}x_\ell,
\]
where the coefficients \((b_{i\ell})\) are uniquely identified.
\end{lemma}

\subsection{Truthful communication in clique partitions}

The previous lemma describes behavior once truthful communication is imposed. We now
ask when truthful communication is itself incentive compatible. A simple sufficient
condition is that agents who communicate with one another have the same local
information. Clique components satisfy exactly this property.

\begin{definition}[Clique partition]
A network \(g\) is a clique partition if every connected component of \(g\) is a clique.
Equivalently, for every connected component \(C\subseteq N\), we have \(ij\in g\) for
all distinct \(i,j\in C\), and there are no links across components.
\end{definition}

In a clique component, every agent observes the same profile of signals under truthful
reporting. Thus agents have the same posterior belief and choose the same action. This
removes the informational asymmetry that would otherwise make a sender want to distort
her report.

\begin{proposition}[Truthful communication in clique partitions]
\label{prop:clique_truthful}
Fix a network \(g\) and a precision profile \(\tau\). If \(g\) is a clique partition,
then the continuation game induced by \((g,\tau)\) admits a truthful continuation
equilibrium.
\end{proposition}

\subsection{Failure of truthful communication in a line}

The previous result shows that clique partitions support truthful communication.
The next example shows that truthful communication need not be sustainable in
non-clique networks.

Consider three agents \(N=\{1,2,3\}\) and the line network
\[
g=\{12,23\}.
\]
Let
\[
\sigma_\theta=\tau_1=\tau_2=\tau_3=\gamma=1.
\]
Then
\[
\bar N_1(g)=\{1,2\},\qquad
\bar N_2(g)=\{1,2,3\},\qquad
\bar N_3(g)=\{2,3\}.
\]

Suppose, toward a contradiction, that truthful communication is an equilibrium.
Under truthful communication, the unique linear action equilibrium is
\[
a_1=\frac{3}{10}x_1+\frac{7}{20}x_2,
\]
\[
a_2=\frac{1}{5}x_1+\frac{3}{10}x_2+\frac{1}{5}x_3,
\]
and
\[
a_3=\frac{7}{20}x_2+\frac{3}{10}x_3.
\]
These coefficients are obtained from the first-order conditions in
Lemma~\ref{lem:linear_truthful}. For example,
\[
\mathbb E[\theta\mid x_1,x_2]=\frac{x_1+x_2}{3},
\qquad
\mathbb E[x_3\mid x_1,x_2]=\frac{x_1+x_2}{3},
\]
so agent \(1\)'s best response is
\(
a_1
=
\frac12\mathbb E[\theta\mid x_1,x_2]
+
\frac12\mathbb E[a_2\mid x_1,x_2]
=
\frac{3}{10}x_1+\frac{7}{20}x_2
\).
The other equations are verified analogously.

Now consider agent \(1\)'s report to agent \(2\). Since strategies have perfect
recall, a deviation by agent \(1\) may involve both the message sent to agent
\(2\) and the subsequent action chosen by agent \(1\). Let agent \(1\) send an
arbitrary message \(m_{12}\). Given agent \(2\)'s equilibrium action strategy,
agent \(2\)'s action is
\[
a_2(m_{12})
=
\frac15 m_{12}
+
\frac{3}{10}x_2
+
\frac15 x_3.
\]
After sending \(m_{12}\) and receiving \(x_2\) from agent \(2\), agent \(1\)'s
best response solves
\[
\min_a
\mathbb E\left[
(a-\theta)^2+(a-a_2(m_{12}))^2
\mid x_1,x_2,m_{12}
\right].
\]
Hence
\[
a_1^{BR}(m_{12})
=
\frac12\mathbb E[\theta\mid x_1,x_2]
+
\frac12\mathbb E[a_2(m_{12})\mid x_1,x_2].
\]
Using
\[
\mathbb E[\theta\mid x_1,x_2]=\frac{x_1+x_2}{3},
\qquad
\mathbb E[x_3\mid x_1,x_2]=\frac{x_1+x_2}{3},
\]
we obtain
\[
a_1^{BR}(m_{12})
=
\frac15 x_1+\frac{7}{20}x_2+\frac{1}{10}m_{12}.
\]
Notice that, when \(m_{12}=x_1\), this coincides with the truthful-equilibrium
action \(a_1=(3/10)x_1+(7/20)x_2\).

At the communication stage, agent \(1\) chooses the message rule as a function of
\(x_1\), anticipating the optimal continuation action chosen later after observing
the message received from agent \(2\). Substituting the optimal continuation action
\(a_1^{BR}(m_{12})\), agent \(1\)'s conditional expected loss is
\[
\mathbb E\left[
\left(a_1^{BR}(m_{12})-\theta\right)^2
+
\left(a_1^{BR}(m_{12})-a_2(m_{12})\right)^2
\mid x_1
\right].
\]
A direct calculation gives
\[
\mathbb E\left[
\left(a_1^{BR}(m_{12})-\theta\right)^2
+
\left(a_1^{BR}(m_{12})-a_2(m_{12})\right)^2
\mid x_1
\right]
=
\frac{16m_{12}^2-40m_{12}x_1+25x_1^2+310}{800}.
\]
The first-order condition therefore gives
\[
m_{12}^*(x_1)=\frac54 x_1.
\]
Thus truthful reporting,
\(
m_{12}(x_1)=x_1,
\)
is not optimal whenever \(x_1\neq 0\). This profitable deviation already allows
agent \(1\) to choose her optimal continuation action after the deviating
message. Therefore truthful communication is not an equilibrium in the line
network.

\section{Assortative Information Clubs}
\label{sec:formation}

We now study which truthful communication structures survive coalitional deviations.
By Proposition~\ref{prop:clique_truthful}, clique partitions support truthful
communication. We therefore focus on truthful clique outcomes and show that the
core selects assortative information clubs.

Throughout this section, order agents so that
\[
\tau_1\geq \tau_2\geq \cdots \geq \tau_n.
\]
No strict ranking is required. If several agents have the same precision, they may
be permuted without changing any payoff comparison.

\begin{definition}[Assortative clique partition]
A clique partition \(\mathcal C=\{C_1,\ldots,C_K\}\) is assortative if its elements
can be ordered so that, whenever \(i\in C_k\), \(j\in C_\ell\), and \(k<\ell\), then
\[
\tau_i\geq \tau_j.
\]
Equivalently, up to permutations of agents with identical precision, each clique is
a consecutive block in the precision ordering.
\end{definition}

For a clique \(C\), let
\[
S_C=\sum_{\ell\in C}\tau_\ell
\]
be the total precision available to its members. Under truthful communication in
\(C\), all members observe the same profile of signals \(x_C\), choose the same
posterior action, and have zero coordination loss. Hence the ex--ante payoff of each
member of \(C\) is
\[
-\frac{1}{\sigma_\theta+S_C}-c(|C|-1).
\]
Thus, among cliques of a fixed size, every member prefers the clique with the
largest total precision.

We construct an assortative clique partition recursively. Let \(r\) be the highest-ranked agent not yet assigned to a clique. Among the
remaining agents, choose an endpoint \(s\ge r\), and form the block
\[
C=\{r,r+1,\ldots,s\},
\]
so as to maximize
\[
-\frac{1}{\sigma_\theta+\sum_{\ell=r}^s\tau_\ell}-c(s-r).
\]
Close this clique, remove its members, and repeat the procedure on the remaining
agents until all agents are assigned. If there are ties, any maximizer may be chosen.
The procedure may also select singleton cliques, in particular for the lowest-precision residual agents when forming additional links may not be payoff improving.

The core argument uses the following individual upper bound. If an agent has
degree \(d\) in any deviating network, then she can receive direct messages from
at most \(d\) other agents, and hence can base her action on at most \(d+1\)
private signals, including her own. If some of these messages are not truthful,
or only partially informative, the informational term cannot improve relative to
the case in which the same signals are truthfully observed. The coordination
term is also always weakly negative and is maximized at zero. Thus any feasible
deviation is bounded above by the ideal case in which the agent observes the
best possible \(d+1\) signals and incurs zero coordination loss.

\begin{proposition}[Core characterization]
\label{prop:core_assortative}
The truthful assortative clique profiles generated by the recursive procedure are
the only outcomes in the core, up to permutations of agents with identical precision.
\end{proposition}

The logic of the result is recursive. For any given number of links, the
highest-precision remaining agents generate the largest attainable informational
benefit, while truthful clique communication eliminates coordination losses.
Hence each clique selected by the recursive procedure gives its members the
largest feasible precision-cost payoff among the remaining agents.

\subsection{Efficiency}

Core-stable truthful clubs need not be socially efficient.\footnote{
The stability notion used here is non-transferable: blocking requires a Pareto
improvement for all members of the deviating coalition, without transfers across
agents. This differs from the utilitarian efficiency notion discussed below, which
aggregates utilities across agents. See \cite{kaneko_wooders_1998}
for the distinction between cooperative games with and without side payments.
} 
The reason is that private stability is based on the payoff of the members of
each club, while social welfare aggregates the informational gains generated by
all signals across all agents.

For a truthful clique partition \(\mathcal C\), aggregate expected payoff is
\[
W(\mathcal C)
=
\sum_{C\in\mathcal C}
|C|
\left(
-\frac{1}{\sigma_\theta+S_C}
-c(|C|-1)
\right).
\]
The function \( 1/(\sigma_\theta+S)
\)
is strictly convex and decreasing in $S$, so the marginal value of precision is higher
in low-precision cliques than in already high-precision cliques. Therefore,
concentrating the most precise agents in the same clique can be privately stable
but socially dominated by a more mixed partition that spreads high-precision
signals across groups.

To see this explicitly, consider four agents: two high-precision agents \(h,h\)
and two low-precision agents \(\ell,\ell\), with \(h>\ell>0\). Let linking costs
be linear,
\[
c(d)=\kappa d.
\]
The recursive procedure first forms the high-precision pair \(\{h,h\}\) rather
than a singleton or a larger clique whenever
\[
\frac{\ell}{(\sigma_\theta+2h)(\sigma_\theta+2h+\ell)}
<
\kappa
<
\frac{h}{(\sigma_\theta+h)(\sigma_\theta+2h)}.
\]
After removing the two high-precision agents, the two low-precision agents form
the pair \(\{\ell,\ell\}\) rather than remaining singletons whenever
\[
\kappa
<
\frac{\ell}{(\sigma_\theta+\ell)(\sigma_\theta+2\ell)}.
\]
Thus the assortative partition
\[
\mathcal C^A=\big\{\{h,h\},\{\ell,\ell\}\big\}
\]
is generated by the recursive procedure whenever
\[
\frac{\ell}{(\sigma_\theta+2h)(\sigma_\theta+2h+\ell)}
<
\kappa
<
\min\left\{
\frac{h}{(\sigma_\theta+h)(\sigma_\theta+2h)},
\frac{\ell}{(\sigma_\theta+\ell)(\sigma_\theta+2\ell)}
\right\}.
\]
For such values of \(\kappa\), the truthful core-stable partition is assortative.

Now compare it with the mixed partition
\[
\mathcal C^M=\big\{\{h,\ell\},\{h,\ell\}\big\}.
\]
The aggregate expected payoff under the assortative partition is
\[
W(\mathcal C^A)
=
-\frac{2}{\sigma_\theta+2h}
-\frac{2}{\sigma_\theta+2\ell}
-4\kappa,
\]
whereas under the mixed partition it is
\[
W(\mathcal C^M)
=
-\frac{4}{\sigma_\theta+h+\ell}
-4\kappa.
\]
The link costs are identical and cancel out in the welfare comparison. Since
\(1/(\sigma_\theta+S)\) is strictly convex in $S$ and \(h>\ell\),
\[
\frac12
\left(
\frac{1}{\sigma_\theta+2h}
+
\frac{1}{\sigma_\theta+2\ell}
\right)
>
\frac{1}{\sigma_\theta+h+\ell}.
\]
Therefore
\[
W(\mathcal C^M)>W(\mathcal C^A).
\]
The assortative truthful partition is core-stable, but it is socially dominated
by the mixed partition.

Thus the same assortative force that supports truthful communication and core
stability can generate inefficient informational segregation.

\section{Conclusion}
\label{sec:conclusion}

This paper studies how the incentive constraints of communication shape endogenous
network formation. Links do not simply transmit information. They also determine
whether information can be transmitted truthfully. A network that gives agents access
to valuable signals may nevertheless fail to support truthful reporting if it creates local
informational asymmetries.

Clique components provide a simple solution. Within a clique, all members observe the
same local profile of signals and choose the same posterior action. Coordination losses
inside the clique are therefore eliminated, and no agent can gain by distorting a report
to another member. Cliques are thus not only dense communication groups; they are
structures that make truthful communication self-enforcing.

With heterogeneous signal precisions and convex linking costs, this incentive constraint
generates assortative information clubs. High-precision agents are valuable partners, but
additional signals have decreasing marginal value. As a result, agents sort into cliques
ordered by precision, and each clique expands only up to the point at which the marginal
informational gain no longer justifies the additional linking cost.

The stable truthful architecture need not be socially efficient. Since each agent does not
internalize the informational value of her signal for other members, equilibrium clubs may
be too small. Moreover, concentrating high-precision agents in the same clique can be
privately stable while a more mixed partition would improve aggregate welfare by spreading
precise signals across groups.

The analysis deliberately keeps the stability notion demanding. Coalitional deviations may
change both the communication network and the continuation strategies available to the
deviating coalition. This is what makes the core characterization sharp: stable networks
must be robust not only to alternative links, but also to alternative truthful communication
regimes. Weaker notions of stability, such as deviations by single agents or pairs, would
enlarge the set of stable architectures and could sustain additional, possibly more efficient,
clique partitions. These alternatives are interesting, but they answer a different question.
The present paper focuses on the benchmark in which truthful communication, network
formation, and coalitional stability are disciplined by the same equilibrium object.

The broader message is that exclusivity in communication networks need not reflect only
access, trust, or status. It may also be an equilibrium response to strategic communication
constraints. Exclusive information clubs can arise because closed groups make truthful
information transmission possible, and because robustness to alternative truthful
communication structures selects assortative clubs.

\bibliographystyle{ecta}
\bibliography{biblio}

\newpage

\appendix
\setcounter{equation}{0}

\begin{center}
	{\LARGE \textbf{Appendix}}
\end{center}

\section{Proofs}\label{app:proofs}

\subsection*{Proof of Proposition~\ref{prop:babbling}}

\begin{proof}
Consider the following strategy profile. For every link \(ij\in g\), agent \(i\)
sends a constant message,
\[
m_{ij}(x_i)=\bar m
\quad \text{for all } x_i,
\]
for some fixed \(\bar m\in\mathbb R\).

Beliefs are such that, after observing any profile of messages and her own signal
\(x_i\), agent \(i\) forms posterior beliefs about \(\theta\) using only \(x_i\).
That is, messages are treated as uninformative about the state, both on and off
the equilibrium path.

Given these beliefs, action strategies are chosen as the Bayesian equilibrium of the
action stage in which each agent conditions only on her own signal.
Thus actions do not depend on messages. Since a unilateral deviation in the
communication stage does not affect any agent's action, it does not affect the
deviator's payoff. Hence no deviation in messages is profitable. Optimality of
actions follows from Bayesian updating based on \(x_i\) alone.

Therefore, the described strategies and beliefs form a Perfect Bayesian equilibrium.
\end{proof}

\subsection*{Proof of Lemma~\ref{lem:linear_truthful}}

\begin{proof}
Under truthful communication, agent \(i\) observes
\[
I_i(g)=\{x_\ell:\ell\in\bar N_i(g)\}.
\]
For \(d_i(g)>0\), agent \(i\)'s first-order condition is
\[
a_i
=
\frac{1}{1+\gamma}\mathbb E[\theta\mid I_i(g)]
+
\frac{\gamma}{1+\gamma}\frac{1}{d_i(g)}
\sum_{j\in N_i(g)}
\mathbb E[a_j\mid I_i(g)].
\]
If \(d_i(g)=0\), the first-order condition is
\[
a_i=\mathbb E[\theta\mid x_i].
\]

Since the prior mean is zero and signals are Gaussian,
\[
\mathbb E[\theta\mid I_i(g)]
=
\frac{\sum_{\ell\in\bar N_i(g)}\tau_\ell x_\ell}
{\sigma_\theta+\sum_{\ell\in\bar N_i(g)}\tau_\ell}.
\]

Let \(\mathcal H_i\) be the space of measurable action functions with finite second moments,
measurable with respect to \(I_i(g)\), and let
\[
\mathcal H=\prod_{i\in N}\mathcal H_i
\]
with norm
\[
\|a\|=\max_{i\in N}\|a_i\|_2.
\]
Define the best-response operator \(T:\mathcal H\to\mathcal H\) by
\[
(Ta)_i
=
\frac{1}{1+\gamma}\mathbb E[\theta\mid I_i(g)]
+
\frac{\gamma}{1+\gamma}\frac{1}{d_i(g)}
\sum_{j\in N_i(g)}
\mathbb E[a_j\mid I_i(g)]
\]
when \(d_i(g)>0\), and by
\[
(Ta)_i=\mathbb E[\theta\mid x_i]
\]
when \(d_i(g)=0\).

For \(d_i(g)>0\),
\[
(Ta)_i-(Ta')_i
=
\frac{\gamma}{1+\gamma}\frac{1}{d_i(g)}
\sum_{j\in N_i(g)}
\mathbb E[a_j-a'_j\mid I_i(g)].
\]
By Jensen's inequality,
\[
\|(Ta)_i-(Ta')_i\|_2
\leq
\frac{\gamma}{1+\gamma}\frac{1}{d_i(g)}
\sum_{j\in N_i(g)}
\|a_j-a'_j\|_2 .
\]
For isolated agents, the left-hand side is zero. Hence
\[
\|Ta-Ta'\|
\leq
\frac{\gamma}{1+\gamma}\|a-a'\|.
\]
Since \(\gamma/(1+\gamma)<1\), \(T\) is a contraction. By the Banach fixed-point
theorem, \(T\) has a unique fixed point in \(\mathcal H\). This fixed point is the
unique Bayesian equilibrium of the action stage.

It remains to show that this equilibrium is linear. Let
\(\mathcal L_i\subseteq\mathcal H_i\) be the finite-dimensional subspace of linear
functions of \((x_\ell)_{\ell\in\bar N_i(g)}\), and let
\[
\mathcal L=\prod_{i\in N}\mathcal L_i.
\]
The operator \(T\) maps \(\mathcal L\) into itself because posterior expectations of
jointly Gaussian linear random variables are linear. Since the fixed point of \(T\)
in \(\mathcal H\) is unique, and since \(T\) has a fixed point in the invariant
subspace \(\mathcal L\), the unique equilibrium belongs to
\(\mathcal L\). Therefore each action is linear and can be written as
\[
a_i=\sum_{\ell\in\bar N_i(g)}b_{i\ell}x_\ell .
\]
Uniqueness of the fixed point implies that the coefficients \((b_{i\ell})\) are
uniquely identified.
\end{proof}

\subsection*{Proof of Proposition~\ref{prop:clique_truthful}}

\begin{proof}
Suppose \(g\) is a clique partition and let \(C\) be a connected component of
\(g\). If \(|C|=1\), there is no communication link to check. Hence consider a
component \(C\) with \(|C|\geq 2\).

Since \(C\) is a clique, for every \(i\in C\),
\[
\bar N_i(g)=C.
\]
Thus, under truthful communication, every agent in \(C\) observes the same
profile of signals
\[
x_C=(x_\ell)_{\ell\in C}.
\]
Define
\[
\mu_C(x_C)
=
\mathbb E[\theta\mid x_C]
=
\frac{\sum_{\ell\in C}\tau_\ell x_\ell}{\sigma_\theta+\sum_{\ell\in C}\tau_\ell}.
\]
If all messages are truthful, all agents in \(C\) choose \(\mu_C(x_C)\). This is
optimal because the coordination term is zero and \(\mu_C(x_C)\) minimizes
\[
\mathbb E[(a-\theta)^2\mid x_C].
\]

We now verify incentive compatibility, allowing for full strategies with perfect
recall. Fix an agent \(i\in C\). Suppose all agents other than \(i\) report
truthfully, and allow agent \(i\) to deviate both in the entire vector of
messages sent to her neighbors and in her subsequent action.

Let
\[
m_i=(m_{ij})_{j\in C\setminus\{i\}}
\]
be an arbitrary vector of messages sent by \(i\). Given the action strategies of
the other agents, neighbor \(j\in C\setminus\{i\}\) chooses
\[
\mu_C(m_{ij},x_{C\setminus\{i\}})
=
\frac{\tau_i m_{ij}+\sum_{\ell\in C\setminus\{i\}}\tau_\ell x_\ell}{\sigma_\theta+\sum_{\ell\in C}\tau_\ell}.
\]
Equivalently,
\[
\mu_C(m_{ij},x_{C\setminus\{i\}})
=
\mu_C(x_C)
+
\frac{\tau_i}{\sigma_\theta+\sum_{\ell\in C}\tau_\ell}(m_{ij}-x_i).
\]
Define
\[
\delta_{ij}
=
\frac{\tau_i}{\sigma_\theta+\sum_{\ell\in C}\tau_\ell}(m_{ij}-x_i).
\]
Thus neighbor \(j\)'s action is \(\mu_C(x_C)+\delta_{ij}\).

After sending \(m_i\) and receiving truthful messages from all neighbors, agent
\(i\) observes \(x_C\) and remembers \(m_i\). Conditional on this information,
her expected loss, up to the posterior variance term
\(\operatorname{Var}(\theta\mid x_C)\), is
\[
\left(a-\mu_C(x_C)\right)^2
+
\frac{\gamma}{|C|-1}
\sum_{j\in C\setminus\{i\}}
\left(a-\mu_C(x_C)-\delta_{ij}\right)^2 .
\]
Let
\[
t=a-\mu_C(x_C).
\]
The continuation problem is therefore equivalent to minimizing
\[
t^2
+
\frac{\gamma}{|C|-1}
\sum_{j\in C\setminus\{i\}}
(t-\delta_{ij})^2
\]
with respect to \(t\).

Under truthful reporting, \(\delta_{ij}=0\) for every \(j\in C\setminus\{i\}\).
Choosing \(t=0\) then makes the displayed expression equal to zero. Since the
expression is nonnegative for every \(t\), this is the global minimum.

If at least one message is non-truthful, then \(\delta_{ij}\neq 0\) for at least one
neighbor \(j\). In that case, no value of \(t\) can make all squared terms equal
to zero, because the first term requires \(t=0\), while the \(j\)-th coordination
term requires \(t=\delta_{ij}\neq 0\). Hence the minimized value is strictly positive.

Since \(\tau_i>0\), \(\delta_{ij}=0\) is equivalent to \(m_{ij}=x_i\). Therefore the
unique optimal message vector is
\[
m_{ij}=x_i
\qquad\text{for every }j\in C\setminus\{i\}.
\]

Thus truthful reporting gives agent \(i\) the lowest possible continuation loss,
even when she is allowed to choose her subsequent action optimally after
remembering the whole vector of messages she sent. Any non-truthful message to
at least one neighbor strictly increases the minimized continuation loss.

Therefore no joint deviation in messages and action is profitable for agent
\(i\). Since \(i\) was arbitrary, truthful reporting is sequentially optimal for
every agent in every clique component. Together with the action rule
\(\mu_C(x_C)\), this yields a truthful continuation equilibrium.
\end{proof}

\subsection*{Proof of Proposition~\ref{prop:core_assortative}}

\begin{proof}
Fix the precision ordering
\[
\tau_1\geq \tau_2\geq \cdots \geq \tau_n .
\]
Let
\[
\mathcal C^*=\{C_1,\ldots,C_K\}
\]
be a clique partition generated by the recursive procedure. 
Throughout the proof, fix one clique partition generated by the recursive procedure.
If there are multiple maximizers at some step, the argument applies to any chosen
maximizer. Permutations among agents with identical precision do not affect the
result.

We first establish a useful bound. Consider any agent \(i\) in any feasible
deviation, and suppose that \(i\) has degree \(d\) in the deviating network. Then
\(i\) can receive direct messages from at most \(d\) other agents and therefore can
base her action on at most \(d+1\) private signals, including her own. If some of
these messages are not truthful, or are only partially informative, the
informational term cannot improve relative to the case in which the same signals
are truthfully observed. Moreover, the coordination term is always weakly negative
and is maximized at zero. Hence \(i\)'s ex-ante expected payoff in any such
deviation is bounded above by the payoff she would obtain if she observed the best
possible \(d+1\) signals and incurred zero coordination loss.

We now apply this bound recursively.

First consider the first clique \(C_1\). Write
\[
C_1=\{1,\ldots,q_1\},
\]
up to permutations of agents with identical precision. Under the truthful
continuation equilibrium from Proposition~\ref{prop:clique_truthful}, all members
of \(C_1\) observe the same signal profile \(x_{C_1}\), choose the same posterior
action, and incur zero coordination loss. Hence every member of \(C_1\) obtains
the ex-ante expected payoff
\[
u(C_1)
=
-\frac{1}{\sigma_\theta+\sum_{\ell=1}^{q_1}\tau_\ell}
-
c(q_1-1).
\]

Now take any member \(i\in C_1\) and any feasible deviation in which \(i\) has
degree \(d\). By the bound above, \(i\)'s ex-ante expected payoff in that deviation
is at most
\[
-\frac{1}{\sigma_\theta+\sum_{\ell=1}^{d+1}\tau_\ell}
-
c(d),
\]
because the largest total precision attainable with \(d+1\) signals is the sum of
the \(d+1\) highest precisions. By construction, \(q_1\) maximizes
\[
-\frac{1}{\sigma_\theta+\sum_{\ell=1}^{q}\tau_\ell}
-
c(q-1)
\]
over all feasible clique sizes \(q\). Therefore, for every \(d\),
\[
-\frac{1}{\sigma_\theta+\sum_{\ell=1}^{d+1}\tau_\ell}
-
c(d)
\leq
-\frac{1}{\sigma_\theta+\sum_{\ell=1}^{q_1}\tau_\ell}
-
c(q_1-1)
=
u(C_1).
\]
Thus no member of \(C_1\) can obtain an ex-ante expected payoff strictly larger
than \(u(C_1)\) in any feasible deviation.

Now consider any candidate core outcome that does not contain a clique generating
the same precision-cost combination as \(C_1\), up to permutations of agents with
identical precision. Then at least one member of \(C_1\) must obtain an ex-ante
expected payoff strictly below \(u(C_1)\). The coalition \(C_1\) can then deviate,
form the clique \(C_1\), and use the truthful continuation equilibrium from
Proposition~\ref{prop:clique_truthful}. This weakly improves every member of
\(C_1\) and strictly improves at least one member. Hence every core outcome must
contain \(C_1\), up to payoff-equivalent permutations of agents with identical
precision.

Remove \(C_1\) from the population and consider the residual set of agents. Let
\(C_2\) be the clique selected by the recursive procedure in this residual
population. The same argument applies. For any member \(i\in C_2\), any feasible
deviation that does not involve members of \(C_1\) gives \(i\) access to at most
\(d+1\) signals if \(i\)'s degree is \(d\). The largest total precision attainable
with \(d+1\) such signals is obtained by the \(d+1\) highest-precision residual
agents. Since \(C_2\) is chosen to maximize the corresponding precision-cost
tradeoff in the residual population, no member of \(C_2\) can obtain an ex-ante
expected payoff strictly larger than \(u(C_2)\). Therefore, if a candidate core
outcome contains \(C_1\) but not \(C_2\), the coalition \(C_2\) blocks it by
forming the truthful clique \(C_2\).

Iterating this argument, every core outcome must contain
\[
C_1,C_2,\ldots,C_K,
\]
up to payoff-equivalent permutations of agents with identical precision. Hence
every core outcome coincides with the recursive assortative clique partition.

It remains to show that the recursive truthful assortative clique profile is
unblocked. Suppose, toward a contradiction, that some coalition \(B\subseteq N\)
blocks \(\mathcal C^*\). Let \(i\) be the highest-ranked agent in \(B\), and let
\(C_k\) be the clique of \(\mathcal C^*\) containing \(i\).

At the stage at which \(C_k\) was selected, all agents with precision strictly
larger than \(\tau_i\) had already been assigned to earlier cliques. Since \(i\) is
the highest-ranked member of \(B\), the blocking coalition cannot use any of those
earlier agents. Thus any deviation available to \(B\) that affects \(i\)'s payoff
uses only agents in the residual population considered at step \(k\).

Suppose that \(i\) has degree \(d\) in the blocking deviation. By the same upper
bound as above, \(i\)'s ex-ante expected payoff in the deviation is at most
\[
-\frac{1}{\sigma_\theta+\sum_{\ell=1}^{d+1}\tau^k_\ell}
-
c(d),
\]
where
\[
\tau^k_1\geq \tau^k_2\geq \cdots
\]
are the precisions in the residual population at step \(k\). Since the recursive
procedure selected \(C_k\) to maximize this expression over all feasible degrees,
\(i\) cannot obtain an ex-ante expected payoff strictly larger than the one she
receives in \(C_k\).

Thus, in any blocking deviation, \(i\) must be exactly indifferent. Exact
indifference can occur only if the deviation gives \(i\) the same
precision-cost combination selected at step \(k\): the same number of links and
the same maximal total precision among the corresponding number of residual
signals. By the tie-breaking convention, this means reproducing \(C_k\), up to
permutations of agents with identical precision. But then no lower-ranked member
of \(B\) can be made strictly better off through this deviation without changing
the precision-cost combination faced by \(i\) or moving \(i\) below her
equilibrium payoff. This contradicts the assumption that \(B\) blocks.

Therefore no coalition blocks the recursive truthful assortative clique profile.
Together with the first part of the proof, this establishes that the recursive
truthful assortative clique profiles are exactly the core outcomes, up to
permutations of agents with identical precision.
\end{proof}

    \end{document}